\documentclass[12pt,english]{article}

\usepackage{fullpage}
\usepackage{amsmath, amssymb, amsthm}
\usepackage{enumerate}
\usepackage{fullpage}
\usepackage{xr}
\usepackage{bbold}
\usepackage{mathtools}
\usepackage[export]{adjustbox}
\usepackage{graphicx}
\usepackage{setspace}
\usepackage{color}
\usepackage{float}
\usepackage{pst-all}
\usepackage{natbib}
\usepackage{etoolbox}
\usepackage{comment}
\usepackage{framed}
\usepackage{xcolor}
\usepackage{tikz}
\newcommand{\lightercolor}[3]{
    \colorlet{#3}{#1!#2!white}
}
\lightercolor{gray}{40}{shadecolor}
 \definecolor{mygray}{gray}{0.8}
 \definecolor{mylightblue}{gray}{0.6}

\usepackage[T1]{fontenc}
\usepackage{babel}
\usepackage{diagbox}
\usepackage[normalem]{ulem}
\usepackage[bottom]{footmisc}
\usepackage{breakcites}
\usepackage[breaklinks=true]{hyperref}

\newcommand{\td}{\text d}

\DeclareMathOperator*{\argmax}{arg\,max}
\DeclareMathOperator*{\minimize}{\mathrm{minimize}}

\newcommand{\bfr}{\mathbf{R}}

\theoremstyle{remark}
\newtheorem{remark}{Remark}
\newenvironment{proofof}[1]{\begin{proof}[Proof of #1]}{\end{proof}}

\usepackage{accents}

\theoremstyle{plain}
\newtheorem{theorem}{Theorem}[section]
\newtheorem{lemma}[theorem]{Lemma}

\newtheorem{corollary}[theorem]{Corollary}

\newtheorem{claim}{Claim}

\theoremstyle{definition}

\newcommand{\fp}{\mathrm{FP}}

\newcommand{\rgrt}{\mathrm{RGRT}}
\newcommand{\dstr}{\mathrm{DSTR}}
\newcommand{\opt}{\mathrm{OPT}}
\newcommand{\cs}{\mathrm{CS}}
\newcommand{\cip}{\mathrm{CIP}}
\newcommand{\wcr}{\mathrm{WCR}}
\title{Robust Monopoly Regulation}
\author{Yingni Guo\thanks{Department of Economics, Northwestern University, yingni.guo@northwestern.edu.} \and Eran Shmaya\thanks{MEDS, Kellogg School of Management, Northwestern University, e-shmaya@kellogg.northwestern.edu.}}
\date{\today}
\begin{document}

\maketitle

\begin{abstract}

We study the regulation of a monopolistic firm using a robust-design approach. We solve for the policy that minimizes the regulator's worst-case regret, where the regret is the difference between his complete-information payoff and his realized payoff. When the regulator's payoff is consumers' surplus, it is optimal to impose a price cap. The optimal cap balances the benefit from more surplus for consumers and the loss from underproduction. When his payoff is consumers' surplus plus the firm's profit, he offers a piece-rate subsidy in order to mitigate underproduction, but caps the total subsidy so as not to incentivize severe overproduction.



\end{abstract}

\section{Introduction}
Regulating monopolies is challenging. A monopolistic firm has the market power to set its price above that in an oligopolistic or competitive market. For instance, \cite{Cooper2018} show that prices at monopoly hospitals are $12\%$ higher than those in markets with four or five rivals. In order to protect consumer well-being, a regulator may want to constrain the firm's price. However, a price-constrained firm may fail to obtain enough revenue to cover its fixed cost, so may end up not producing. The regulator must balance the need to protect consumer well-being and the need to not distort the production.

This challenge could be solved easily if the regulator had complete information about the industry. The regulator could ask the firm to produce at the socially optimal level and to set price equal to the marginal cost. He could then subsidize the firm for all of its other costs. However, the regulator typically has limited information about the consumer demand or the technological capacity of the firm. How shall the regulatory policy be designed when the regulator knows considerably less about the industry than the firm does? If the regulator demands robustness and wants a policy that works ``fairly well'' in all circumstances, what shall this policy look like?

We address this classic problem of monopoly regulation (e.g., \cite{BaronMyerson1982}) with a non-Bayesian approach. The regulator's payoff is a weighted sum of consumers' surplus and the firm's profit. He can regulate the firm's price and/or quantity. He can give a subsidy to the firm or charge a tax from it. Given a policy, the firm chooses its price and quantity to maximizes its profit. The \emph{regret} to the regulator is the difference between what he could have gotten if he had complete information about the industry and what he actually gets. The regulator evaluates a policy by its worst-case regret, i.e., the maximal regret he can incur across all possible demand and cost scenarios. The optimal policy minimizes the worst-case regret. 

The worst-case regret approach to uncertainty is our most significant difference from \cite{BaronMyerson1982} and the literature on monopoly regulation in general. \cite{BaronMyerson1982} take a Bayesian approach to uncertainty by assigning a prior to the regulator over the demand and cost scenarios and characterizing the policy that minimizes the expected regret. (Minimizing the expected regret is the same as maximizing the expected payoff, since the regulator's expected complete-information payoff is constant.) We instead focus on industries where information asymmetry is so pronounced that there is no obvious way to formulate a prior, or industries where new sources of uncertainty arise all the time. (See \cite{Hayek1945}, \cite{Weitzman1974} and \cite{Carroll2019}, for instance, for elaboration of these points.) In response, the regulator looks for a robust policy that works fairly well in all circumstances. 

To illustrate our solution, we begin with two extreme cases of the regulator's payoff. If the regulator puts no weight on the firm's profit, so his payoff is consumers' surplus, then it is optimal to impose a price cap. A price cap bounds how much consumers' surplus that the firm can extract. Consumers benefit from a lower price. However, a price cap might discourage a firm which should have produced from producing. Consumers lose in this case due to the firm's underproduction. The optimal level of the price cap balances consumers' gain from a lower price and their loss from the firm's underproduction. 

If the regulator puts the same weight on the firm's profit as he does on consumers' surplus, so his payoff is the total surplus of consumers and the firm, then the regulator simply wants the firm to produce as efficiently as possible. Given that an unregulated monopolistic firm tends to supply less than the efficient level, the regulator wants to encourage more production by subsidizing the firm. However, subsidy might incentivize production above the efficient level. The optimal design of subsidy must balance the loss from underproduction and that from overproduction.

The regulator will have a target price and a subsidy cap. For each unit that the firm sells, he subsidizes the firm for the difference between its price and the target price, subject to the constraint that the total subsidy doesn't exceed the subsidy cap. This piece-rate subsidy up to the target price effectively lifts the firm's selling price, motivating the firm to serve more than just those consumers with high values. On the other hand, the cap on the firm's total subsidy makes sure that the regulator doesn't lose too much from the potential overproduction. 


For intermediate payoffs, the regulator puts some weight on the firm's profit, but this weight is lower than the weight he puts on consumers' surplus. He must balance three goals simultaneously: giving more surplus to consumers, mitigating underproduction and mitigating overproduction. It is optimal to combine the price cap and the subsidy rule described above, leading to a regulatory policy with three distinctive features. First, the regulator will impose a price cap so the firm can't get more than the price cap per unit. Second, the firm gets a piece-rate subsidy instead of a lump-sum one. Third, the firm is subsidized up to the price cap subject to a cap on the total subsidy it will get. As the regulator puts more weight on the firm's profit, the level of this price cap increases.




Our contribution is threefold. First, we solve for an optimal regulatory policy. Second, our result explains why price cap regulation and piece-rate subsidy are common in practice. Third, we introduce the worst-case regret approach to the regulation problem. We advocate for this approach over the Bayesian approach for the two shortcomings of the Bayesian approach that are also emphasized in \cite{ArmstrongSappington2007}. First, since the relevant information asymmetries can be difficult to characterize precisely, it is not clear how to formulate a prior. Second, since multi-dimensional screening problems are difficult to solve, the form of optimal regulatory policies is generally not known.


\paragraph{Related literature.}

This paper contributes to the literature on monopoly regulation. \cite{Caillaud1988} and \cite{Braeutigam1989} provide an overview of earlier contributions in this field. \cite{ArmstrongSappington2007} discuss the recent developments. Our paper is closely related to \cite{BaronMyerson1982}. The most significant difference is our approach to uncertainty. \cite{BaronMyerson1982} take a Bayesian approach to uncertainty, and assume that there is a one-dimensional cost parameter that is unknown to the regulator. We take a non-Bayesian, worst-case regret approach, and assume that the regulator lacks information about both the demand and the cost functions. 


Our paper contributes to the growing literature of mechanism design with worst-case objectives. \cite{Carroll2019} provides a survey of recent theory in this field. Most of this literature assumes that the designer aims to maximize his worst-case payoff. We assume that the designer aims to minimize his worst-case regret. From this aspect, we are closely related to \cite{HurwiczShapiro1978}, \cite{BergemannSchlag2008, BergemannSchlag2011}, \cite{RenouSchlag2011}, and \cite{Bevia2019}. \cite{HurwiczShapiro1978} show that a 50-50 split is an optimal sharecropping contract when the optimality criterion involves the ratio of the designer's payoff to the first-best total surplus. \cite{BergemannSchlag2008, BergemannSchlag2011} examine robust monopoly pricing and argue that minimizing the worst-case regret is more relevant than maximizing the worst-case payoff, since the latter criterion suggests pricing to the lowest-value buyer. \cite{RenouSchlag2011} apply the solution concept of $\varepsilon$-minimax regret to the problem of implementing social choice correspondences. \cite{Bevia2019} characterize contests in which contestants have dominant strategies and find within this class the contest for which the designer's worst-case regret is minimized.

Minimizing the worst-case regret is more relevant a criterion in our setting as well for two reasons. First, the regret in our setting has a natural interpretation: it is the weighted sum of distortion in production and the firm's profit. Second, the regulator's worst-case payoff is zero or less under any policy, since consumers' values might be too low relative to the cost. In this case, there is no surplus even under complete information. When there is no surplus, there is nothing the regulator can do. We argue that, instead, the regulator's goal should be to protect surplus in situations where there is some surplus to protect. The notion of regret catches this idea.


The worst-case regret approach goes back at least to \cite{Savage1954}. Under this approach, when a decision maker has to choose some action facing uncertainty, he chooses the action that minimizes the worst-case regret across all possible realizations of the uncertainty. The regret is defined as the difference between what the decision maker could achieve given the realization, and what he achieves under this action. In our case, the regulator has uncertainty about the demand and cost functions and he has to choose a policy. Savage also puts forward an interpretation of the worst-case regret approach in the context of group decision, which is relevant for our policy design context.  Consider a group of people who must jointly choose a policy. They have the same payoffs but different probability judgements. 
Under  the policy that minimizes the worst-case regret no member of the group faces a large regret, so no member will feel that the suggestion is a serious mistake.
Seminal game theory papers in which players try to minimize worst-case regret include \cite{Hannan1957} and \cite{HartMasColell2000}. Minimizing worst-case regret is also the leading approach in online learning, and in particular in multi-armed bandit problems (see \cite{bandit2012} for a survey). 


Our work also contributes to the delegation literature (e.g., \cite{Holmstrom1977,Holmstrom1984}). When the regulator cares only about consumers' surplus, it is optimal to simply impose a price cap. To our knowledge, we are the first to show that a delegation contract --- a contract that doesn't use money --- is optimal in a contracting environment where both parties can transfer money to each other. 

\section{Environment}
There is a monopolistic firm and a mass one of consumers. Let $V:[0,1]\to [0,\bar v]$ be a decreasing upper-semicontinuous \emph{inverse-demand function}. A quantity-price pair $(q,p)\in [0,1]\times [0,\bar v]$ is \emph{feasible} if and only if it is below the inverse-demand function, i.e., $p\leqslant V(q)$. The total value to consumers of quantity $q$ is the area under the inverse-demand function, given by $\int_0^q V(z) ~\td z$.

Let $C:[0,1]\to \bfr_+$ with $C(0)=0$  be an increasing lower-semicontinuous \emph{cost function}. 
The \emph{social optimum} is given by: \begin{equation}\label{opt}\opt = \max_{q \in [0,1]} \left( \int_0^q V(z) ~\td z - C(q)\right).\end{equation} 
If the firm produces $q$ units, then the \emph{(market) distortion} is given by: \begin{equation}\label{dstr}\dstr= \opt - \left(\int_0^q V(z)~\td z - C(q)\right).\end{equation}
To simplify notation, we omitted the dependence of $\opt$ on $V,C$ and the dependence of $\dstr$ on $V,C$ and $q$. We will do the same for some other terms in this section when no confusion arises.


\subsubsection*{Regulatory policies} A \emph{policy} is given by an upper-semicontinuous function $\rho:[0,1]\times [0,\bar v]\to \bfr$. If the firm sells $q$ units at price $p$, then it receives revenue $\rho(q,p)$. The firm's revenue from the policy, $\rho(q,p)$, includes the revenue $qp$ from the marketplace, and any tax or subsidy, $\rho(q,p)-qp$, imposed by the regulator. We give three examples of regulatory policies: 
\begin{enumerate}
	\item A regulator who decides not to intervene will choose $\rho(q,p)=qp$, so the firm's revenue $\rho(q,p)$ equals its revenue from the marketplace. 
	\item The regulator can give the firm a lump-sum subsidy $s$ if it sells more than a certain quantity $\tilde q$. The policy is $\rho(q,p)=qp$ if $q < \tilde q$ and $\rho(q,p)=qp+s$ if $q \geqslant \tilde q$.
	\item The regulator can require that the firm get no more than $k$ per unit. The policy is $\rho(q,p)=\min(qp,qk)$. If the firm prices above $k$, it pays a tax of $q(p-k)$ to the regulator. 
\end{enumerate}

Fix a policy $\rho$, an inverse-demand function $V$ and a cost function $C$. If the firm produces $q$ units at price $p$, then \emph{consumers' surplus} and the \emph{firm's profit} are given by:
\begin{equation}\label{fpcs}
\cs= \int_0^q V(z)~\td z-\rho(q,p) , \text{ and }\fp = \rho(q,p) - C(q).\end{equation}
The definition of consumers' surplus incorporates the fact that any subsidy to the firm is paid by consumers through their taxes and that any tax from the firm goes to the consumers. We also assume that $\rho(0,0) \geqslant 0$, so the firm is allowed to stay out of business without suffering a negative profit. This is the participation constraint.

We say that $(q,p)$ is a \emph{firm's best response to $(V,C)$ under the policy $\rho$} if it maximizes the firm's profit over all feasible $(q,p)$. The firm might have multiple best responses. The participation constraint implies that $\fp\geqslant 0$ for every best response $(q,p)$ of the firm.

The regulator's payoff is a weighted sum, $\cs+\alpha \fp$, of consumers' surplus and the firm's profit for some fixed parameter $\alpha \in [0,1]$.
\subsubsection*{Complete information}
Fix an inverse-demand function $V$ and a cost function $C$. We let $\cip$ denote the regulator's complete-information payoff, which is  what the regulator would achieve if he could tailor the policy for these inverse-demand and cost functions. Formally, 
\begin{equation}\label{cip}
\cip=\max_{\rho, q, p} \left( \cs+\alpha \fp \right),
\end{equation}
where the maximum ranges over all policies $\rho$ and all firm's best responses $(q,p)$ to $(V,C)$ under $\rho$.

The following claim shows that the regulator's complete-information payoff is the social optimum. He would ask the firm to produce the socially optimal quantity and give the firm a revenue equal to its cost. As a result, the maximum surplus is generated, all of which goes to consumers. 
\begin{claim}\label{cl:cip}Fix an inverse-demand function $V$ and a cost function $C$. Then $\cip=\opt$.
\end{claim}
\begin{proof}
First, the regulator's complete-information payoff is at most $\opt$. Indeed, 
$$
\cs+\alpha \fp \leqslant \cs+\fp \leqslant \opt,
$$	
for every policy $\rho$ and every best response $(q,p)$ to $(V,C)$ under $\rho$. Here the first inequality follows from $\alpha\leqslant 1$ and the participation constraint $\fp\geqslant 0$, and the second from the definitions of $\cs,\fp,\opt$ in \eqref{fpcs} and \eqref{opt}.

Second, let $q^\ast$ denote a quantity that achieves the social optimum. The regulator can achieve $\opt$ by setting
\[\rho(q,p)=\begin{cases}C(q^\ast), &\text{ if }(q,p)=(q^\ast, V(q^\ast)), \\0, &\text{ otherwise}.\end{cases}\]Choosing $(q,p)=(q^\ast, V(q^\ast))$ is a firm's best response to $(V,C)$ under $\rho$ which gives $\cs=\opt$ and $\fp=0$, so $\cs+\alpha \fp=\opt$. 
\end{proof}
\subsubsection*{Regret}
When the regulator does not know $(V,C)$, the policy will usually not give the regulator his complete-information payoff. The regulator's \emph{regret} is the difference between what he could have gotten under complete information and what he actually gets. The following step allows us to express regret in terms of distortion and the firm's profit:
\begin{align*}
\rgrt & =  \cip-\left( \cs +\alpha \fp \right) = \opt-\left( \cs +\alpha \fp \right)  \\
& = \opt-\left( \cs + \fp \right)+ (1-\alpha) \fp  \\
&  = \dstr+ (1-\alpha) \fp.
     \end{align*}
Here, the first equality follows from the definition of regret, the second from Claim \ref{cl:cip} that $\cip =\opt$, and the rest is algebra. 

Regret has a natural interpretation in our setting. $\dstr$ represents the loss in efficiency, since the regulator wishes the firm to produce as efficiently as possible. $(1-\alpha)\fp$  represents the loss in his redistribution objective, since the regulator wishes that more surplus goes to consumers rather than to the firm.

\subsubsection*{The regulator's problem}
We look for the policy that minimizes the worst-case regret. Thus \emph{the regulator's problem} is
\[\minimize_\rho \; \max_{V,C,q,p} \; \rgrt\]
where the minimization is over all policies $\rho$, and the maximum ranges over all $(V,C)$ and all the firm's best responses $(q,p)$ to $(V,C)$ under $\rho$. 

Formulating the regulator's problem as a minimax problem is our only departure from the literature on monopoly regulation. If we assigned a Bayesian prior to the regulator over the demand and cost scenarios, minimizing the expected regret would be the same as maximizing the expected payoff as in \cite{BaronMyerson1982}. We instead consider environments where the regulator knows only the range of consumers' values, which is much easier to figure out than formulating a prior.   

\begin{remark}
In the definition of $\cip$ we assumed that the firm breaks ties in favor of the regulator, whereas in the definition of the regulator's problem we assumed that the firm breaks ties against the regulator. These assumptions are for convenience only and do not affect the value of $\cip$ in Claim~\ref{cl:cip} or the solution to the regulator's problem in Theorems~\ref{th:lower} to \ref{th:properties}.\footnote{If in the definition of $\cip$ we assumed that
  the firm breaks ties against the regulator, we would define $$ \cip = \sup_\rho \min_{q,p}  (\cs + \alpha \fp ),$$
  where the minimum ranges over all firm's best responses $(q,p)$ to $(V,C)$ under $\rho$. Then the supremum may not be achieved, but the value of $\cip$ would be the same. Similarly, if we assumed that the firm breaks ties in favor of the regulator in the regulator's problem then the ``worst case'' pair $(V,C)$ may not exist, but the solution to the regulator's problem would remain the same.} \end{remark}
\section{Main result}


We first provide a lower bound on the worst-case regret of any policy. We then show that our policy indeed achieves this lower bound, so it is optimal. Both the lower-bound and upper-bound discussions  center on the tradeoff between giving more surplus to consumers, mitigating underproduction, and mitigating overproduction.

Suppose that the regulator imposes a price cap $k$. A price cap advances the regulator's redistribution objective by bounding how much consumers' surplus the firm can extract, but it may worsen the problem of underproduction. There is a price cap level that balances these two forces. Explicitly, consider a market in which every consumer has the highest value $\bar v$. If the cost is zero, the firm will price at $k$ and serve all consumers. There is no distortion since all consumers are served, as it should be, but the firm's profit is $k$. The regret is $(1-\alpha)k$. The lower $k$ is, the lower the regret is. On the other hand, if the firm has a fixed cost of $k$, it is a firm's best response not to produce. The firm's profit is zero, but the distortion is $\bar v -k$, which is the surplus that could have been made. The regret equals this distortion. The lower $k$ is, the higher the regret is. We let $k_\alpha$ be the price cap such that these two levels of regret are equalized, so $k_\alpha=\bar v/(2-\alpha)$ as depicted in the left panel of Figure \ref{fg:lower}. 


%
%

With this $k_\alpha$ balancing the tradeoff between giving more surplus to consumers and mitigating underproduction, we are ready to establish a lower bound on the worst-case regret. 

\begin{theorem}[Lower bound on worst-case regret] \label{th:lower}
Let
  \begin{equation}
  \label{eq:regretlower}
r_\alpha  = \max_{q\in [0,1], \;   p\in [0,k_\alpha]} \min\left( (1-\alpha)q k_\alpha - qp\log q,  q(k_\alpha-p)\right).
\end{equation}
Then the worst-case regret under any policy is at least $r_\alpha$.
\end{theorem}
For any $(q,p)$, we argue that the worst-case regret is at least  the minimum of two terms. Roughly speaking, the first term, $(1-\alpha)q k_\alpha - qp\log q $, is the possible regret from underproduction if the revenue to the firm is too low. The second term, $q(k_\alpha-p)$, is the possible regret from overproduction if the revenue is too high. No matter how the policy is designed, the regulator has to suffer from one of these two. Since the worst-case regret is at least the minimum of these two for every $(q,p)$, we can take the maximum over $q\in[0,1]$ and $p \in [0, k_\alpha]	$.

Let $q_\alpha$ achieve the maximum in the definition of $r_\alpha$ in~\eqref{eq:regretlower}. When $\alpha\leqslant 1/2$, $q_\alpha$ equals one. When $\alpha>1/2$, $q_\alpha$ is interior. The explicit values of $r_\alpha$ and $q_\alpha$ are given by:
\begin{equation*}
\label{eq:lowerbound}
r_\alpha=\bar v \begin{cases}
 \frac{1-\alpha}{2-\alpha} &  \alpha \leqslant \frac{1}{2} \\
 \frac{\left(2+\alpha -\sqrt{\alpha  (\alpha +4)}\right) e^{1-\frac{\alpha+\sqrt{\alpha  (\alpha +4)}}{2} }}{2
   (2-\alpha)} &  \alpha > \frac{1}{2}.
\end{cases},\quad \quad q_\alpha=\begin{cases}1,&\text{  if }\alpha\leqslant 1/2\\ e^{1-\frac{\alpha+\sqrt{\alpha  (\alpha +4)}}{2} }, & \text{ if }\alpha>1/2.
\end{cases}
	\end{equation*}
The middle and right panels of Figure \ref{fg:lower} depict the values of $r_\alpha$ and $q_\alpha$. 
\begin{figure}[htb!]
\begin{center}
\psset{xunit=3,yunit=3}
\begin{pspicture}(-.08,-.1)(5.2,1.3)

\psline[arrows=->](0,0)(1.1,0)
\psline[arrows=->](0,0)(0,1.1)
\rput[c](-.05,-.1){$0$}
\rput[c](1,-.1){$1$}
\rput[c](-.08,1){$\bar v$}
\rput[c](-.08,.5){$\frac{\bar v}{2}$}
\rput[c](1.18,0){$\alpha$}
\rput[c](0.5,0.786667){$k_\alpha$}
\psline[linewidth=1.5pt,linecolor=black](0., 0.5)(0.05, 0.512821)(0.1, 0.526316)(0.15, 0.540541)(0.2,0.555556)(0.25, 0.571429)(0.3,0.588235)(0.35,0.606061)(0.4, 0.625)(0.45, 0.645161)(0.5,0.666667)(0.55,0.689655)(0.6, 0.714286)(0.65, 0.740741)(0.7, 0.769231)(0.75, 0.8)(0.8, 0.833333)(0.85, 0.869565)(0.9, 0.909091)(0.95, 0.952381)(1., 1.)
\psline[linestyle=dotted](1,0)(1,1)(0,1)

\psline[arrows=->](2,0)(3.1,0)
\psline[arrows=->](2,0)(2,1.1)
\rput[c](1.95,-.1){$0$}
\rput[c](3,-.1){$1$}
\rput[c](1.92,1){$\bar v$}
\rput[c](3.18,0){$\alpha$}
\rput[c](1.92,.5){$\frac{\bar v}{2}$}
\psline[linestyle=dotted](3,0)(3,1)(2,1)
\rput[c](2.5, 0.413333){$r_\alpha$}

\psline[linewidth=1.5pt,linecolor=black](2., 0.5)(2.05, 0.487179)(2.1, 0.473684)(2.15, 0.459459)(2.2,
   0.444444)(2.25, 0.428571)(2.3, 0.411765)(2.35, 
  0.393939)(2.4, 0.375)(2.45, 0.354839)(2.5, 0.333333)(2.55, 
  0.312508)(2.6, 0.294179)(2.65, 0.277976)(2.7, 
  0.263611)(2.75, 0.250852)(2.8, 0.239517)(2.85, 
  0.229457)(2.9, 0.220557)(2.95, 0.212721)(3., 0.205881)

\psline[arrows=->](4,0)(5.1,0)
\psline[arrows=->](4,0)(4,1.1)
\rput[c](3.95,-.1){$0$}
\rput[c](5,-.1){$1$}
\rput[c](3.92,1){$1$}
\rput[c](5.18,0){$\alpha$}
\psline[linestyle=dotted](5,0)(5,1)(4,1)
\rput[c](4.75, 0.847072){$q_\alpha$}
\rput[c](4.5,-.1){$1/2$}
\psline[linestyle=dotted](4.5,0)(4.5,1)

\psline[linewidth=1.5pt,linecolor=black](4., 1.)(4.05, 1.)(4.1, 1.)(4.15, 1.)(4.2, 1.)(4.25, 
  1.)(4.3, 1.)(4.35, 1.)(4.4, 1.)(4.45, 1.)(4.5, 1.)(4.55,
   0.936164)(4.6, 0.877514)(4.65, 0.823438)(4.7, 
  0.773432)(4.75, 0.727072)(4.8, 0.684001)(4.85, 0.64391)(4.9,
   0.606531)(4.95, 0.571631)(5., 0.539003)

\psline[linewidth=1.5pt,linecolor=black,linestyle=dashed](2., 0.)(2.05, 0.025641)(2.1, 0.0526316)(2.15, 
  0.0810811)(2.2, 0.111111)(2.25, 0.142857)(2.3, 
  0.176471)(2.35, 0.212121)(2.4, 0.25)(2.45, 0.290323)(2.5, 
  0.333333)

\psline[linewidth=1.5pt,linecolor=black,linestyle=dashed](2.5, 
  0.313333)(2.55, 
  0.292508)(2.6, 0.274179)(2.65, 0.257976)(2.7, 
  0.243611)(2.75, 0.230852)(2.8, 0.219517)(2.85, 
  0.209457)(2.9, 0.200557)(2.95, 0.192721)(3., 0.185881)
  
  \rput[c](2.3, 0.072857){$s_\alpha$}

\end{pspicture}
\end{center}
\caption{Values of $k_\alpha$, $r_\alpha$, $q_\alpha$, and $s_\alpha$}
\label{fg:lower}
\end{figure}
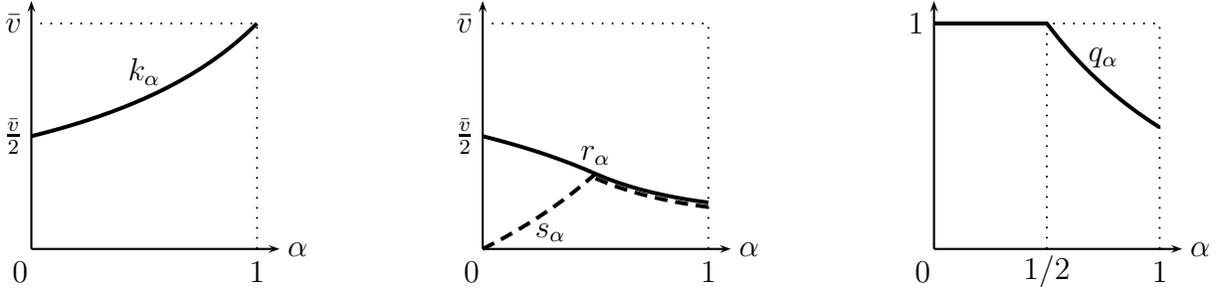

\begin{theorem}[Optimal policy] \label{th:upper}
Let 
$$
s_\alpha=\sup\{q(k_\alpha - p): q\in [0,1], p\in [0,k_\alpha], (1-\alpha)q k_\alpha - qp\log q> r_\alpha\}.
$$
The policy 
\begin{equation}
\label{eq:optimalform}
  \rho(q,p)=\min( q k_\alpha , qp  + s )
  \end{equation}
with $s_\alpha\leqslant s\leqslant r_\alpha$ achieves the worst-case regret $r_\alpha$.
\end{theorem} 
We first provide some intuition as to how a policy of the form \eqref{eq:optimalform} addresses the three goals of giving more surplus to consumers, mitigating underproduction, and mitigating overproduction simultaneously. First, the firm can't get more than $k_\alpha$ for each unit it sells. This caps how much consumers' surplus that the firm can extract. Second, a monopolistic firm has the tendency to serve just those consumers with very high values. In order to incentivize the firm to produce more, the regulator subsidizes the firm for the difference between its price and $k_\alpha$. This piece-rate subsidy effectively increases the firm's selling price to $k_\alpha$. Third, the firm's total subsidy is capped by $s$, so the potential overproduction induced by subsidy is also under control.

Depending on how much the regulator cares about the firm's profit, he puts different weights on these three goals, and hence varies $k_\alpha$ and $s$ as $\alpha$ varies. 

The explicit value of $s_\alpha$ is given below, and is depicted as the dashed line in the middle panel of Figure \ref{fg:lower}. 
\begin{equation*}
 s_\alpha=
\begin{cases}
\bar v \frac{\alpha}{2-\alpha} &  \alpha \leqslant \frac{1}{2} \\
 r_\alpha &  \alpha > \frac{1}{2}.
\end{cases}
	\end{equation*}
Note that $s_\alpha=0$ when $\alpha=0$. Hence, the policy $\rho(q,p)=q \min( \bar v/2, p )$ is optimal when $\alpha=0$, and it simply requires that the firm get less than $\bar v/2$ per unit. 


The optimal policy in Theorem~\ref{th:upper} features three properties. First, the fact that $\rho(q,p)\leqslant  q k_\alpha$ for every $q$ implies a price cap: The firm cannot get more than $k_\alpha$ per unit sold. To see the price cap more explicitly, consider the  policy $\tilde\rho$ given by:
\begin{equation*}
 \tilde\rho(q,p)=
\begin{cases}
\rho(q,p) &  p\leqslant k_\alpha \\
 -\infty &  p>k_\alpha.
\end{cases}
	\end{equation*}
The policy $\tilde\rho$ induces similar behavior to that of $\rho$ in the sense that $(q,p)$ is a best response to $\rho$ if and only if $(q,\min(p,k_\alpha))$ is a best response to $\tilde\rho$ and these responses give the same consumers' surplus. Therefore, by Theorem~\ref{th:upper}, $\tilde\rho$ is also optimal. The second property is that for some quantity-price pairs the total subsidy to the firm is at least $s_\alpha$. The third property is that the total subsidy to the firm is at most $r_\alpha$. Theorem~\ref{th:properties} asserts that every optimal policy has similar properties. Recall that $q_\alpha$ achieves the maximum in the definition of $r_\alpha$ in~\eqref{eq:regretlower}.
\begin{theorem}\label{th:properties}
  Let $\rho$ be an optimal policy. Then
  \begin{enumerate}
  \item (Price cap): $\rho(q,p)\leqslant q k_\alpha $ for every $q\leqslant q_\alpha$.
  \item (Subsidy): There exists some $(q,p)$ such that $\rho(q,p) \geqslant qp + s_\alpha$.
  \item (Subsidy cap): $\rho(q,p) \leqslant qp + r_\alpha$ for every $(q,p)$.
  \end{enumerate}
\end{theorem}
In particular, since $q_\alpha=1$ for $\alpha\leqslant 1/2$, it follows from Theorem~\ref{th:properties} that for $\alpha\leqslant 1/2$ a price cap at $k_\alpha$ is necessary for every level of production.


\section{Discussions}

\subsection{Incorporating additional knowledge}
\label{re:costs}
In our model we made no assumptions on the inverse-demand or the cost functions except for monotonicity, semicontinuity, and the range of consumers' values (between $0$ and $\bar v$). The regulator may know more than this. We can extend our framework in an obvious way to incorporate the regulator's knowledge by restricting the set of inverse-demand and cost functions in the regulator's problem. For instance, the regulator may know that the firm has a constant marginal cost together with a fixed cost, but doesn't know these cost levels. This is the type of cost functions used most frequently in studies of monopoly regulation. 

In our proof of Theorem \ref{th:lower}, we establish a lower bound on the worst-case regret of any policy using only fixed cost functions, i.e., $C(q)$ is constant for any $q>0$. (See remark \ref{rm:fixedcost} for details.) This means that Theorem \ref{th:lower} remains true for every set of cost functions that includes the set of all fixed cost functions. Once we know that Theorem \ref{th:lower} remains true, we know that our policy in Theorem~\ref{th:upper} remains optimal, since the worst-case regret under our policy is at most $r_\alpha$ across all inverse-demand and cost functions.

Of course, the regulator may do better than $r_\alpha$ if he has significant knowledge about the industry. We believe that incorporating the regulator's additional knowledge is an exciting direction for future research, which will demonstrate the adaptability of the worst-case regret approach. Our analysis and policy serve as the very first step toward understanding the optimal policy for any particular industry. 


\subsection{The efficient rationing assumption}

In our model we allow the firm not to clear the market, and we assume that if this happens then the consumers who are being served are the ones with higher values. Indeed, absent some additional assumptions on the cost function, even a firm which operates under a price cap may prefer not to clear the market. 

A common assumption in the monopoly regulation literature is that the firm has decreasing average cost, i.e., the average cost $C(q)/q$ is decreasing for $q>0$. Since the set of all fixed cost functions satisfies the decreasing average property, by subsection~\ref{re:costs} Theorem \ref{th:lower} remains correct under this decreasing average assumption on the cost function, and our policy in Theorem \ref{th:upper} is optimal. Moreover, if the cost function satisfies this assumption, then a firm which operates under our policy will want to clear the market.



\section{Proofs}

\subsection{Preliminaries}
For every $q,p$ we let $V_{q,p}$ and $W_{q,p}$ be the inverse-demand functions given by: \[V_{q,p}(z)=\begin{cases} \bar v, &\text{ if }z\leqslant q\\ \frac{qp}{z}, &\text{ if }q<z\leqslant 1,\end{cases}\text{  and   }W_{q,p}(z)=\begin{cases}p, &\text{ if }z\leqslant q\\0, &\text{ if }q<z\leqslant 1,\end{cases}\]
as shown in Figure \ref{fg:prelim}. The inverse demand $W_{q,p}$ has the property that, among all inverse-demand functions under which $(q,p)$ is feasible, $W_{q,p}$ generates the least total value to consumers. 

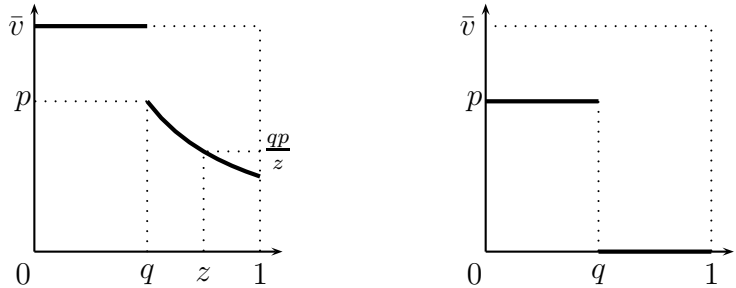
\begin{figure}[htb!]
\begin{center}
\psset{xunit=3,yunit=3}
\begin{pspicture}(-.1,-.1)(3.2,1.1)

\psline[arrows=->](0,0)(1.1,0)
\psline[arrows=->](0,0)(0,1.1)
\rput[c](-.05,-.1){$0$}
\rput[c](1,-.1){$1$}
\rput[c](-.08,1){$\bar v$}
\psline[linestyle=dotted](1,0)(1,1)(0,1)
\rput[c](.5,-.1){$q$}
\rput[c](-.05,.66667){$p$}
\psline[linestyle=dotted](0,.666667)(.5,.666667)(.5,0)

\psline[linestyle=dotted](0.75,0)(0.75, 0.444444)(1,0.444444)
\rput[c](0.75,-.1){$z$}
\rput[c](1.08,0.444444){$\frac{q p}{z}$}

\psline[linewidth=1.6pt,linecolor=black](0,1)(0.5,1)
\psline[linewidth=1.6pt,linecolor=black](0.5, 0.666667)(0.5625, 0.592593)(0.625, 0.533333)(0.6875, 
  0.484848)(0.75, 0.444444)(0.8125, 0.410256)(0.875, 
  0.380952)(0.9375, 0.355556)(1., 0.333333)


\psline[arrows=->](2,0)(3.1,0)
\psline[arrows=->](2,0)(2,1.1)
\rput[c](1.95,-.1){$0$}
\rput[c](3,-.1){$1$}
\rput[c](2.5,-.1){$q$}
\psline[linestyle=dotted](2.5,0)(2.5,0.666667)
\rput[c](1.92,1){$\bar v$}
\rput[c](1.95,0.666667){$p$}
\psline[linestyle=dotted](3,0)(3,1)(2,1)
\psline[linewidth=1.6pt,linecolor=black](2,0.666667)(2.5,0.666667)
\psline[linewidth=1.6pt,linecolor=black](2.5,0)(3,0)

\end{pspicture}
\end{center}
\caption{$V_{q,p}$ and $W_{q,p}$ demand}
\label{fg:prelim}
\end{figure}

To understand the role of $V_{q,p}$ in our argument, consider an unregulated firm (i.e., a firm which operates under the policy $\rho(q,p)=qp$). If the inverse-demand function is $V_{q,p}$ and the  cost is zero, then selling $q$ units at price $\bar v$ is a firm's best response. This response causes distortion of $\int_{q}^1 V_{q,p}(z)~\td z=-qp\log q$ due to underproduction. The following lemma shows that this is the worst distortion that can happen when the firm is unregulated.

\begin{lemma}\label{le:thelemma}
Assume that an unregulated firm sells $\bar q$ units at a price $\bar p$ such that $\bar p \geqslant \sup_{z>\bar q}{V(z)}$ . Let
\[
\opt_{\bar q} = \max_{q\geqslant \bar q} \int_{\bar q}^q V(z)~\textup{d} z - (C(q)-C(\bar q))	 
\]
be the maximal additional surplus to society if the firm has produced $\bar q$ units,
and let
\[\fp_{\bar q} = \max_{q\geqslant \bar q} \;  q  \min(\bar p, V(q)) - \bar q\bar p -(C(q)- C(\bar q)),\]
be the maximal additional profit to the firm if it has produced $\bar q$ units and commits to price at most $\bar p$.
Then \[\opt_{\bar q} \leqslant \fp_{\bar q} + \varphi(\bar q) \bar p,\]
where \[\varphi(q)=\begin{cases} \frac{1}{e}, &\text{ if } q<\frac{1}{e} \\-q\log q,&\text{ if }q\geqslant \frac{1}{e},\end{cases}\] is the least decreasing majorant of $q \mapsto -q\log q$.\end{lemma}
The lemma does not assume that selling $\bar q$ units or more is optimal for the firm. Therefore, the assertion in the lemma still holds even if the best response for an unregulated firm is to produce less than $\bar q$ units at a possibly higher price than $\bar p$.
\begin{proofof}{Lemma~\ref{le:thelemma}}
  We can assume that $\fp_{\bar q}=0$, otherwise replace $C$ with $\tilde C$ such that $\tilde C(z)=C(z)$ if $z\leqslant {\bar q}$ and $\tilde C(z)= C(z)+\fp_{\bar q}$  if  $z>{\bar q}$.
  
  Let $q'\in\argmax_q  \int_{\bar q}^{q} V(z)~\td z - (C(q)-C({\bar q}))$. Let $c'=C(q')-C({\bar q})$. 

  Since the firm does not want to produce more, it follows that $z V(z)-C(z)\leqslant {\bar q}\bar p - C({\bar q})$ for every $z>\bar q $, so that \[V(z)\leqslant \frac{{\bar q}\bar p+C(z)-C({\bar q})}{z}  \leqslant \frac{{\bar q}\bar p+c'}{z},\] for $\bar q < z \leqslant q'$.
  
Since $V(z)\leqslant \bar p$ for $z>\bar q$ it follows from the definition of $q'$ that $c'\leqslant (q'-{\bar q})\bar p$.  Let $q{''}= {\bar q}+\frac{c'}{\bar p}$, so $q{''}\leqslant q'$. Then 
\begin{multline*}\int_{\bar q}^{q'} V(z)~\td z - c' \leqslant (q{''}-{\bar q})\bar p + \int_{q{''}}^{q'} \frac{{\bar q}\bar p+c'}{z}~\td z- c' \leqslant (q{''}-{\bar q})\bar p + \int_{q{''}}^{1}\frac{{\bar q}\bar p+c'}{z}~\td z- c' \\ =\int_{q{''}}^{1}\frac{{\bar q}\bar p+c'}{z}~\td z= -q{''} \bar p  \log q{''} \leqslant \varphi({\bar q})\bar p,\end{multline*}
where the first step uses the fact that $V(z)\leqslant \bar p$ for ${\bar q}< z\leqslant q{''}$, the second to last step follows from ${\bar q}\bar p+c'= q{''}\bar p$, and the last step follows from $\bar q \leqslant q{''}$ and the definition of $\varphi$.
\end{proofof}
For $\bar q=0$ and $\bar p=\bar v$ the lemma has the following corollary which is interesting for its own sake. It bounds from below an unregulated firm's profit in a market with a high social optimum. We are unaware of previous statements of this corollary, but similar arguments to those in the proof of Lemma~\ref{le:thelemma}  with zero cost appeared in \cite{RoeslerSzentes2017}.
\begin{corollary}
  For an unregulated firm which best responds to  $(V,C)$, we have \[\fp \geqslant \opt -  \frac{\bar v}{e}.\]
\end{corollary}

\subsection{Lower bound on worst-case regret}
For a policy $\rho$ let
\[\wcr(\rho)=\; \max_{V,C,q,p} \; \rgrt\]
where the maximum ranges over all $(V,C)$ and all the firm's best responses $(q,p)$ to $(V,C)$ under $\rho$.

For a policy $\rho$ let $\bar\rho(q)=\max_{q'\leqslant q, p'\leqslant \bar v}\rho(q',p')$ be the maximal revenue the firm can get under $\rho$ from selling $q$ units or less, and let $\hat\rho(q,p)=\max_{q'\geqslant q, q'p'\leqslant qp}\rho(q',p')$ be the maximal revenue under $\rho$ if the firm sells at least $q$ units and the revenue from the marketplace is at most $qp$. As shown in Figure \ref{fg:piclowerproof}, $\bar\rho(q)$ is the maximum of $\rho$ in the light-gray area, and $\hat\rho(q,p)$ is the maximum of $\rho$ in the dark-gray area. 
 \begin{figure}[htb!]
\begin{center}
\psset{xunit=3,yunit=3}
\begin{pspicture}(-.1,-.1)(1.1,1.1)

\pscustom[linewidth=0pt,fillstyle=solid,fillcolor=mygray,linecolor=mygray]{ 
\psline(0,0)(.5,0)(.5,1)(0,1)}
\rput[c](.25,.5){\textbf{{\color{black}$\bar \rho(q)$}}}
\pscustom[linewidth=0pt,fillstyle=solid,fillcolor=mylightblue,linecolor=mylightblue]{ 
\psline(.5,0)(0.5, 0.666667)(0.5625, 0.592593)(0.625, 0.533333)(0.6875, 
  0.484848)(0.75, 0.444444)(0.8125, 0.410256)(0.875, 
  0.380952)(0.9375, 0.355556)(1., 0.333333)(1,0)}
\rput[c](0.75,.225){\textbf{{\color{black}$\hat \rho(q,p)$}}}

\psline[arrows=->](0,0)(1.1,0)
\psline[arrows=->](0,0)(0,1.1)
\rput[c](-.05,-.1){$0$}
\rput[c](1,-.1){$1$}
\rput[c](-.08,1){$\bar v$}
\psline[linestyle=dotted](1,0)(1,1)(0,1)
\rput[c](.5,-.1){$q$}
\rput[c](-.05,.66667){$p$}
\psline[linestyle=dotted](0,.666667)(.5,.666667)(.5,0)

\psline[linestyle=dotted](0.75,0)(0.75, 0.14)
\psline[linestyle=dotted](0.75,.28)(0.75, 0.444444)
\psline[linestyle=dotted](0.75, 0.444444)(1,0.444444)
\rput[c](0.75,-.1){$z$}
\rput[c](1.08,0.444444){$\frac{p q}{z}$}




\end{pspicture}
\end{center}
\caption{Definitions of $\bar\rho(q)$ and $\hat\rho(q,p)$}
\label{fg:piclowerproof}
\end{figure}
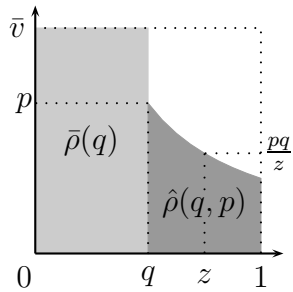

We first show that the worst-case regret under a policy is at least the maximum subsidy that this policy offers. 
\begin{claim}\label{cl:lb-over}Fix a policy $\rho$. Then $\wcr(\rho)\geqslant \rho(q,p)-qp$ for every $q,p$.\end{claim}
\begin{proof}If $\rho(q,p)\leqslant qp$ the assertion follows from the fact that regret is nonnegative. Assume that $\rho(q,p)\geqslant qp$ and consider the inverse-demand function $W_{q,p}$ and a  fixed cost $\rho(q,p)$. Then the firm will produce $q$ units at price $p$, with $\fp=0$ and
  \[\rgrt = \dstr = \rho(q,p) - qp,\]
  because of overproduction.
\end{proof}

We then show that, if the firm doesn't receive sufficiently more revenue from producing more, there is sizable regret due to underproduction.
\begin{claim}\label{cl:underq}Fix a policy $\rho$. Let $\underbar q\leqslant q\in [0,1]$ and let $p\in [0,k_\alpha]$. If $\hat\rho(q,p) \leqslant \bar\rho(\underbar q) + (q-\underbar q)k_\alpha$ then
  \[\wcr(\rho)\geqslant (1-\alpha)(\bar\rho(\underbar q)+(q-\underbar q)k_\alpha) - qp\log q.\]
    \end{claim}
    \begin{proof}\begin{enumerate}
      \item If $\bar\rho(q)-\bar\rho(\underbar q)\leqslant (q-\underbar q)k_\alpha$, then consider the inverse-demand function $V_{q,p}$ and a cost function such that producing $\underbar q$ units or less is costless and producing additional units incurs a fixed cost of $(q-\underbar q)k_\alpha$. The firm will produce at most  $\underbar q$ units, with $\fp=\bar \rho(\underbar q)$ and \[\dstr\geqslant (q-\underbar q)(\bar v - k_\alpha) - qp\log q=(1-\alpha)  (q-\underbar q)k_\alpha-qp\log q,\] because of underproduction. Therefore
  \[\rgrt=(1-\alpha)\fp+\dstr\geqslant 
    (1-\alpha)(\bar \rho(\underbar q)+ (q-\underbar q)k_\alpha) - qp\log q.\]
  \item If $\bar\rho(q)-\bar\rho(\underbar q)\geqslant (q-\underbar q)k_\alpha$, then consider the inverse-demand function $V_{q,p}$ and zero cost.
The firm will produce at most $q$ units, with
  $\fp=\bar\rho(q)\geqslant \bar \rho(\underbar q)+(q-\underbar q)k_\alpha$ and $\dstr\geqslant -qp\log q$ because of underproduction. Therefore \[\rgrt=(1-\alpha)\fp+\dstr \geqslant (1-\alpha)\left(\bar \rho(\underbar q)+ (q-\underbar q)k_\alpha\right) -qp\log q.\]
  \end{enumerate}
\end{proof}

Combining Claims \ref{cl:lb-over} and \ref{cl:underq}, we show that the regulator suffers sizable regret either from underproduction or from overproduction. 
\begin{claim}\label{cl:lb-withunderq}Fix a policy $\rho$. Let $\underbar q\leqslant q\in [0,1]$ and let $p\in [0,k_\alpha]$. Then
  \[\wcr(\rho)\geqslant \min\bigl((1-\alpha)(\bar\rho(\underbar q)+(q-\underbar q)k_\alpha) - qp\log q, \bar\rho(\underbar q) + (q-\underbar q)k_\alpha-qp\bigr).\]
\end{claim}
\begin{proof}
If $\hat\rho(q,p)\geqslant \bar\rho(\underbar q) + (q-\underbar q)k_\alpha$ then let $q',p'$ be such that $q' p'\leqslant q p $, $q'\geqslant q$ and $\rho(q',p')=\hat\rho(q,p)$. By Claim~\ref{cl:lb-over}
\[\wcr(\rho) \geqslant \rho(q',p') - p'q'\geqslant\bar\rho(\underbar q)+(q-\underbar q) k_\alpha- qp.\]

If $\hat\rho(q,p)<\bar\rho(\underbar q) + (q-\underbar q)k_\alpha$ then $\wcr(\rho)\geqslant  (1-\alpha)(\bar\rho(\underbar q)+(q-\underbar q)k_\alpha) - qp\log q$  by  Claim~\ref{cl:underq}.\end{proof}
\subsection{Proof of Theorem~\ref{th:lower}}
We need to show that $\wcr(\rho)\geqslant \min\left( (1-\alpha)q  k_\alpha  - qp\log q, q (k_\alpha-p) \right)$ for every $q,p$. This follows from Claim~\ref{cl:lb-withunderq} with $\underbar q = 0$. 
\begin{remark}\label{rm:fixedcost}
The proof of Claim~\ref{cl:lb-withunderq} for the  case of $\underbar q=0$ relies only on fixed cost functions. So does the proof of  Theorem~\ref{th:lower}. \end{remark}
\subsection{Upper bound on worst-case regret}
We consider a policy of the form \begin{equation}\label{r-k-s}\rho(q,p)=\min(qk, qp+s).\end{equation}
We bound the regret from~\eqref{r-k-s} separately for the case of overproduction and the case of underproduction. 
\begin{claim}\label{cl:over}The regret from overproduction under~\eqref{r-k-s} is at most
  \[\max\left((1-\alpha)k,s\right).\]\end{claim}
\begin{proof} 
Let $q^\ast$ be a socially optimal quantity, let $p^\ast=V(q^\ast)$, and assume that the firm chooses $(q,{p})$ with $q \geqslant q^\ast$ and ${p}\leqslant V(q) \leqslant p^\ast$.
    Let $\bar c = C(q)-C(q^\ast)$. Then
    \begin{equation}\label{dstrover}\dstr = \bar c-\int_{q^\ast}^{q}V(z)~\td z\leqslant \bar c - ({q}-q^\ast){p},\end{equation} and
\begin{equation}\label{barcover}\bar c\leqslant \rho({q},{p})-\rho(q^\ast,p^\ast)\end{equation}
since $({q},{p})$ is a best response.
Therefore
\begin{multline*}\rgrt = (1-\alpha)\fp+\dstr \leqslant
  (1-\alpha)(\rho(q,{p})-\bar c) + \bar c - ({q}-{q}^\ast){p} \leqslant \\ (1-\alpha)\rho(q^\ast,p^\ast) + \rho(q,p)-\rho(q^\ast,p^\ast)-(q-q^\ast)p  \leqslant  (1-\alpha)\rho(q^\ast,p^\ast) + \rho(q,p)-\rho(q^\ast,p)-(q-q^\ast)p \leqslant \\
  (1-\alpha)\rho(q^\ast,p^\ast) + (q - q^\ast) (\rho(q,p) /q - p)   \leqslant   (1-\alpha) q^\ast k + (1 - q^\ast/q) s \leqslant \\ (1-\alpha) q^\ast k + (1 - q^\ast)s\leqslant \max((1-\alpha)k, s).\end{multline*}
where the first inequality follows from~\eqref{fpcs}, the fact that $\bar c\leqslant C({q})$ and~\eqref{dstrover}; the second inequality follows from~\eqref{barcover}; the third inequality follows from the fact that $p\leqslant p^\ast$ and the fact that $p\mapsto \rho(q,p)$ is monotone increasing; in the fourth inequality, $\rho(q,p)-\rho(q^\ast,p) \leqslant (q-q^\ast) \rho(q,p) /q$ because $q^\ast\leqslant q$ and $q\mapsto \rho(q,p) /q$ is decreasing; the fifth inequality follows from $\rho(q^\ast,p^\ast)\leqslant (1-\alpha)  q^\ast k$,  $\rho(q,p)\leqslant qp+s$, and $q\leqslant 1$.

\end{proof}

\begin{claim}\label{cl:under}The regret from underproduction under~\eqref{r-k-s} is at most
\[\max_q  \; q \max((1-\alpha)k, \bar v - k) - q\left(k-\frac{s}{q}\right) \log q.\] 
\end{claim}
\begin{proof}
  Let $q^\ast $ be a socially optimal quantity and assume that the firm chooses $({q},{p})$ with ${q}\leqslant q^\ast$. 

  If $q^\ast(k - V(q^\ast))\leqslant s$ then $\rho(q^\ast,V(q^\ast))= q^\ast k$ and $\rho(q,p)=q k $. Therefore, since the firm prefers to produce $q$ over $q^\ast$ it follows that $C(q^\ast)-C(q)\geqslant (q^\ast-q)k$ which implies that $\dstr\leqslant (q^\ast - q)(\bar v - k)$ and
  \[\rgrt\leqslant (1-\alpha)\rho(q,p)+\dstr\leqslant (1-\alpha)qk + (q^\ast - q)(\bar v - k)\leqslant \max((1-\alpha)k,\bar v - k).\]

  If $q^\ast(k - V(q^\ast)) > s$ then let $\bar q\in [q,q^\ast)$ be such that $z(k - V(z))\leqslant s$ for $q< z < \bar q$ and $z(k-V(z))> s$ for $z>\bar q$. ($\bar q$ is the point at which the subsidy is used up, except that if it was already used up before $q$ then $\bar q=q$). Let $\bar p=k - s/\bar q$. Then it follows from the definition of $\bar q$ that $\bar p\geqslant \sup_{z > \bar q}V(z)$. 

By Lemma~\ref{le:thelemma} there exists some $z^\ast\in [\bar q, q^\ast]$ such that 
\[\int_{\bar q}^{q^\ast} V(z)~\td z - (C(q^\ast)-C(\bar q)) \leqslant z^\ast \underline{p}- \bar q \bar p -(C(z^\ast)- C(\bar q))+ \varphi(\bar q) \bar p,\]
with $\underline{p}=\min(\bar p, V(z^\ast))$. Since $z^\ast\geqslant \bar q$ and $\underline{p}\leqslant \bar p$ it follows from the definition of $\rho$ that
\[\rho\left(z^\ast, \underline{p}\right) - z^\ast \underline{p}  \geqslant \rho(\bar q,\bar p)-\bar q \bar p= \bar q(k - \bar p) \Longrightarrow  z^\ast \underline{p}  -\rho\left(z^\ast, \underline{p}\right)\leqslant \bar q( \bar p-k)  .\]
Since the firm prefers to produce ${q}$ over $z^\ast$ it follows that
\begin{equation*}
\rho\left(z^\ast, \underline{p}\right)\leqslant \rho({q},{p}) + (C(z^\ast)-C({q})).
\end{equation*}
The last three inequalities and $C(\bar q)\leqslant C(q^\ast)$ imply
\begin{equation}\label{barqqast}
\int_{\bar q}^{q^\ast} V(z)~\td z  \leqslant C(q^*)-C(q) + \rho (q,p)-\bar q k +\varphi(\bar q) \bar p .  
\end{equation}
Therefore
\begin{multline*}\dstr=\int_{{q}}^{q^\ast} V(z) ~\td z - (C(q^\ast)-C({q}))=\int_{{q}}^{\bar q} V(z) ~\td z+\int_{\bar q }^{q^\ast} V(z) ~\td z - (C(q^\ast)-C({q})) \\\leqslant  (\bar q - {q})\bar v - \bar q k + \rho({q},{p}) + \varphi(\bar q) \bar p\leqslant (\bar q  - {q})(\bar v - k) + \varphi(\bar q)\bar p,\end{multline*}
where the first inequality follows from~\eqref{barqqast} and $V(z)\leqslant \bar v$, and the second from $\rho({q},{p})\leqslant q k$.  It follows that 
\begin{multline*}\rgrt\leqslant (1-\alpha)\rho({q},p)+\dstr \leqslant (1-\alpha)qk + (\bar q - {q})(\bar v - k) + \varphi(\bar q) \bar p \leqslant \\
q' \max((1-\alpha)k, \bar v - k) - q'\bar p \log q' \leqslant q' \max((1-\alpha)k, \bar v - k) - q'(k - s/q') \log q' ,
\end{multline*}
for some $\bar q\leqslant q' < 1$ such that $\varphi(\bar q)=-q'\log q'$. Here the last inequality follows from the fact that $\bar p \leqslant k - s/\bar q \leqslant k - s/q'$.
\end{proof}

\subsection{Proof of Theorem~\ref{th:upper}}
Note first that from~\eqref{eq:regretlower} with $q=1$ and $p=0$ we get $r_\alpha\geqslant (1-\alpha)k_\alpha=\bar v-k_\alpha$.

Consider the policy~\eqref{r-k-s} with $k=k_\alpha$ and $s_\alpha \leqslant s\leqslant r_\alpha$.

Since $s\leqslant r_\alpha$ and $\bar v -k_\alpha = (1-\alpha) k_\alpha \leqslant r_\alpha$ it follows from Claim~\ref{cl:over} that the regret from overproduction is at most $r_\alpha$.

By Claim~\ref{cl:under} to prove that the regret from underproduction is at most $r_\alpha$ it is sufficient to prove that $(1-\alpha)q k_\alpha - q(k_\alpha-s/q) \log q\leqslant r_\alpha$ for every $q\in [0,1]$. For $q=1$ this follows from the fact that $(1-\alpha)k_\alpha\leqslant r_\alpha$. Let $q<1$ and let $p=k_\alpha-s/q$. Let $q'>q$. Then $q'(k_\alpha - p) > s$ and by the assumption on $s$ this implies that $(1-\alpha)q'k_\alpha - q'p\log q' \leqslant r_\alpha$. Since this is true for every $q'>q$ it follows by continuity that $(1-\alpha)qk_\alpha - qp\log q \leqslant r_\alpha$, as desired.
\subsection{Proof of Theorem~\ref{th:properties}}
Let $(q_\alpha,p_\alpha)$ achieve the maximum in the definition of $r_\alpha$ in~\eqref{eq:regretlower}. 
\begin{enumerate}
\item Assume that $\rho(\underbar q,p)> \underbar q  k_\alpha$ for some $\underbar q\leqslant q_\alpha$ and some $p$. Then $\bar\rho(\underbar q)>\underbar q k_\alpha$ and therefore $\bar\rho(\underbar q)+(q_\alpha-\underbar q)k_\alpha> q_\alpha k_\alpha$. Therefore, by Claim~\ref{cl:lb-withunderq} with $q=q_\alpha$ and $p=p_\alpha$,
  \begin{multline*}\wcr(\rho)\geqslant \min\bigl((1-\alpha)(\bar\rho(\underbar q)+(q_\alpha-\underbar q)k_\alpha) - q_\alpha p_\alpha  \log q_\alpha, \bar\rho(\underbar q) + (q_\alpha-\underbar q)k_\alpha- q_\alpha p_\alpha\bigr)\\> \min((1-\alpha) q_\alpha k_\alpha-q_\alpha p_\alpha \log q_\alpha, q_\alpha (k_\alpha - p_\alpha))=r_\alpha.\end{multline*}
\item Suppose that $\rho(q,p)<qp +s_\alpha$ for every $q,p$. Then there exists some $q\in [0,1], p\in [0,k_\alpha]$ such that $(1-\alpha)q k_\alpha - qp\log q> r_\alpha$ and $q(k_\alpha-p)>\max_{q',p'} (\rho(q',p')-p'q')\geqslant \hat\rho(q,p)-qp$, implying $\hat\rho(q,p)<q k_\alpha $. By Claim~\ref{cl:underq} with $\underbar q=0$ we get that $\wcr(\rho) > r_\alpha$.
  \item Suppose that $\rho(q,p)>qp+r_\alpha$ for some $q,p$. Then $\wcr(\rho)>r_\alpha$ by  Claim~\ref{cl:lb-over} 
\end{enumerate}

\nocite{*}
\bibliographystyle{aea}
\bibliography{cite2}

\end{document}